\documentclass[11pt]{amsart}
  \usepackage{geometry}                
  \usepackage{graphicx}
 \usepackage{epstopdf}
 \usepackage{mydefs}
\usepackage{xcolor}
\usepackage{algorithm}
\usepackage{algorithmic}
\usepackage{fullpage}
\usepackage[foot]{amsaddr}

\newcommand{\cal}[1]{\mathcal{#1}}

\newcommand{\todo}[1]{{\begin{small}\sffamily \color{gray}TODO:  #1 \end{small}}}
\newcommand{\header}[1]{\smallskip\noindent\textbf{#1}\quad}

\newcommand{\BO}{\mathcal{O}}

\newtheorem{claim}{Claim}

\def\squarebox#1{\hbox to #1{\hfill\vbox to #1{\vfill}}}

\newlength{\tablength}
\newlength{\spacelength}
\newcommand{\tabstar}{\hspace*{\tablength}}
\newcommand{\spacestar}{\hspace*{\spacelength}}
\def\obeytabs{\catcode`\^^I=\active}
{\obeytabs\global\let^^I=\tabstar}
{\obeyspaces\global\let =\spacestar}
\newenvironment{display}{\begingroup\obeylines\obeyspaces\obeytabs}{\endgroup}
\newenvironment{prog}{\begin{display}\parskip0pt\sf}{\end{display}}

\def\calL{{\cal L}}
\def\calP{{\cal P}}
\def\calQ{{\cal Q}}

\def\SetR{R}

\newcommand{\prob}[1]{\textsf{#1}}  
\newcommand{\alg}[1]{\textbf{#1}}   

\def\calP{{\cal P}}   
\newcommand{\PCopt}{\overline{OPT}}

\newcommand{\sched}{\prob{PC-Scheduling}}

\newcommand{\powp}{\calP_p}  

\def\fncons2{8 \cdot 3^{\alpha}}

\begin{document}
\title{\Large The Power of Non-Uniform Wireless Power}


%
%

\author{Magn\'us M. Halld\'orsson$^\ast$}
\thanks{$^\ast$ICE-TCS, School of Computer Science,
    Reykjavik University, Reykjavik, Iceland. 
    \texttt{mmh@ru.is, ppmitra@gmail.com}}
\author{Stephan Holzer$^\dag$}
\thanks{$^\dag$Distributed Computing Group, ETH Zurich, Switzerland. \texttt{\{stholzer, wattenhofer\}@ethz.ch}}
\author{Pradipta Mitra$^\ast$} 

\author{Roger Wattenhofer$^\dag$}

\date{}

\begin{abstract}
We study a fundamental measure for wireless interference in the SINR model
known as (weighted) inductive independence.
This measure characterizes the effectiveness of using \emph{oblivious} power --- when 
the power used by a transmitter only depends on the distance to the receiver  --- as
a mechanism for improving wireless capacity. 

We prove optimal bounds for inductive independence, implying a number of algorithmic
applications. An algorithm is provided that achieves --- due to existing lower bounds --- capacity that is asymptotically best possible using oblivious power assignments. Improved
approximation algorithms are provided for a number of problems for oblivious
power and for power control, including distributed scheduling, connectivity, secondary
spectrum auctions, and dynamic packet scheduling.
\end{abstract}

\maketitle

\section{Introduction}

One of the strongest weapons for increasing the capacity of a wireless
network is power control. Higher power increases the bandwidth of a
single transmission link, while causing more interference to other
simultaneously transmitting links.  Given this tension,
intelligent power control is crucial in increasing the spatial reuse
of the available bandwidth. Thus it is not surprising that most
contemporary wireless protocols use some form of power control.
More recently, this phenomenon has also been studied
theoretically; it was shown in a series of works that power
control may improve the capacity of a wireless network in an
exponential \cite{MoWa06,us:esa09full} or even unbounded
\cite{FKRV09} way.

Unrestricted power control is, however, a double-edged sword. 
In order to achieve the
theoretically best results, one must solve complex optimization
problems, where transmission power of one node potentially depends on
the transmission powers of all other nodes \cite{KesselheimSoda11}. In real wireless networks,
where communication demands change over time, this may not be an option. In
practical protocols, the transmission power should be independent of
other concurrent transmissions, which leaves it to only depend on 
the distance between transmitter and receiver. This
is known as \emph{oblivious} power control.

Many questions immediately rise in the wake of the previous assertion:
What is the price of restricting power control to oblivious powers? Which of the infinitely many oblivious power schemes are
good choices? Once an oblivious power scheme is chosen, what algorithmic results
can be achieved?

In this work, we look at these questions in the context of the physical or SINR model
of interference, a realistic model gaining 
increasing
attention (see Section \ref{sec:related} for historical background and motivation and Section \ref{sec:model} for precise definitions). In this setting,
our work answers a number of these questions optimally, completing an extensive
line of work in the algorithmic study of the SINR model.

The specific problem at the center of our work is \emph{capacity} maximization: Given a set of transmission links (each a transmitter-receiver pair), 
find the largest subset of links that can transmit simultaneously.

Before the present work, the state-of-the-art was as follows. The mean power assignment,
where a link of length $\ell$ is assigned power (proportional to) $\ell^{\alpha/2}$ ($\alpha$ being a small
physical constant), had emerged as the ``star'' among oblivious power assignments. It was
shown that using mean power, one can approximate capacity maximization with respect to arbitrary power control 
within a factor of  $\BO(\log n \cdot \log\log\Delta)$
 \cite{us:esa09full} and $\BO(\log n + \log\log \Delta)$ \cite{SODA11}, where $\Delta$ is the ratio between the maximum and minimum
transmission distance and $n$ is the number of links in the system. 
This showed that the somewhat earlier lower bound of $\Omega(n)$ \cite{FKRV09} applied only when $\Delta$ was doubly exponential. In terms of $\Delta$, it was shown
that one \emph{must} pay an $\Omega(\log\log\Delta)$
factor \cite{us:esa09full}. The best upper bounds were, as mentioned, either dependent 
on the size of the input \cite{us:esa09full,SODA11} and as such
unbounded (in relation to $\Delta$), or exponentially worse  ($\BO(\log \Delta)$)
\cite{DBLP:conf/infocom/AndrewsD09,gouss2007a}.


\subsection{Our Contributions}
%

In this paper, we study all power assignments of the form $\ell^{p
  \cdot \alpha}$ for all fixed $0 < p < 1$ (setting $p = \frac12$
results in mean power). Our first result shows that the lower bound of
$\Omega(\log \log \Delta)$ is tight.  That is, we give a simple
algorithm that uses any oblivious power scheme from the above class,
achieving a solution quality within an $\BO(\log\log\Delta)$-factor of the 
optimum with unrestricted power control. This is an asymptotically optimal solution quality for this class of schemes according to \cite{us:esa09full}. For small to moderate values
of $\Delta$, e.g., when $\Delta$ is at most polynomial in $n$ (which
presumably includes most real-world settings), our bound is an
exponential improvement over all previous bounds, including the $\BO(\log
\Delta)$-bound of \cite{DBLP:conf/infocom/AndrewsD09} (see also
\cite{gouss2007a}).

This result
extends the ``star status'' from mean
power to a large class of assignments. This class has been studied
implicitly before in a wide array of work \cite{SODA11,HM11a,DBLP:conf/spaa/HoeferKV11,KV10} on ``length-monotone, sub-linear'' power assignments, but
its relation to arbitrary power was not understood.

Our second main contribution is to 
improve a number of algorithmic results that use these power assignments. 
We shave a logarithmic approximation factor off
a variety of problems, including distributed scheduling \cite{KV10}, secondary spectrum auctions \cite{DBLP:conf/spaa/HoeferKV11}, wireless connectivity \cite{PODC12,HM12,MoWa06}, 
and dynamic packet scheduling \cite{sirocco12,kesselheimStability}.  Using
the capacity relation between oblivious and arbitrary power (our first result), we 
strengthen the bounds for these problems in the power control setting as well.

Though we have presented our work above in terms of algorithmic implications,
what we actually prove are two \emph{structural} results, from which this host
of algorithmic applications follow essentially immediately. 
These results are important in their own right, e.g., implying tight bounds on 
certain efficiently computable measures of interference.

To provide an intuitive understanding of our results, it is useful to
recall the graph theoretic notion of \emph{inductive independence} 
\cite{DBLP:journals/talg/YeB12}.  A graph $G$ has inductive
independence number $d$ if there is an ordering of the vertices $v_1,
v_2, \ldots, v_n$ such that each $v_i$ has at most $d$ edges to any 
independent set $I \subseteq \{v_{i+1}, v_{i+2},\ldots,v_n\}$. An example is provided in Figure \ref{fig:independence}.
\begin{figure*}[ht]
	\begin{center}
		\includegraphics[scale=0.5]{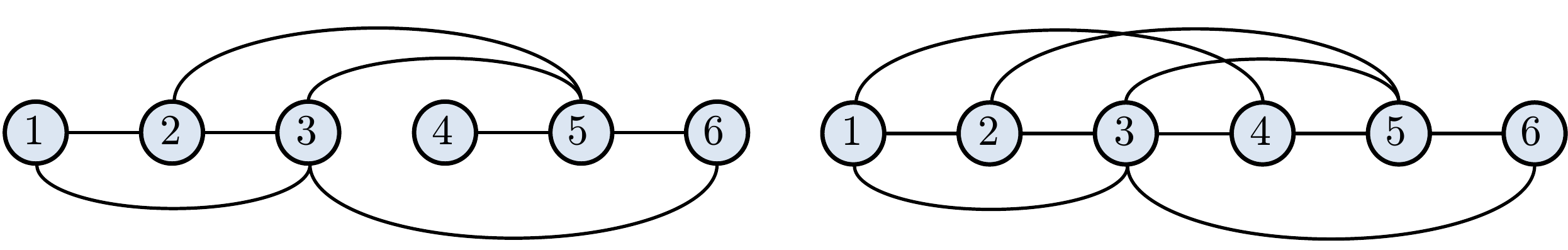}
	\end{center}
	\caption{The graph on the left has inductive independence number $1$, the graph on the right has inductive independence number $2$.}\label{fig:independence}
\end{figure*}
The inductive independence property is found in many graph classes
(e.g., intersection graphs of convex planar objects are 3-inductive
independent \cite{DBLP:journals/talg/YeB12}), and it has powerful algorithmic
implications
\cite{Halldorsson00approximationsof,DBLP:conf/spaa/HoeferKV11,DBLP:journals/talg/YeB12}.
For example, a simple $d$-approximation algorithm for the maximum
independent set problem in such a graph is as follows: Process the
vertices in the prescribed order, adding each vertex to the solution
if it has no edges to nodes already in the solution.  By the inductive
independence property, the addition of a single vertex disqualifies at
most $d$ vertices of the optimal solution from being added in the future,
which implies the claimed approximation factor.

In this paper, we deal with an interference measure that is a natural analog of
inductive independence, applied to certain weighted graphs that model the SINR
interference scenario.
In this context, links are vertices, and the edge weights
represent the extent of interference between links. The relevant ordering
of the links is the ascending order by length, and ``independent sets'' are represented by \emph{feasible sets} of links (links that can transmit simultaneously).

When feasibility is with respect to \emph{arbitrary} power assignment, 
we show that the measure is
bounded by $\BO(\log \log \Delta)$ (Theorem \ref{thm:effective}), implying
our first capacity result (and its applications). Technically, this is
done by carefully extending the analysis of
\cite{us:esa09full}. When feasibility is with respect to oblivious
power from the above mentioned class, the measure can be bounded by
a constant (Theorem \ref{thm:indind}), implying the second set of algorithmic
results. This involves a potentially novel contradiction technique (at least
in the context of SINR analysis).

Our results hold for general metric spaces and all constants $\alpha > 0$.
Apart from the specific applications pinpointed here, we expect any number of
future algorithmic questions in the SINR model to directly benefit
from these bounds. 

\subsection{Related Work}\label{sec:related}
Gupta and Kumar \cite{kumar00} were among the first to provide analytical results for wireless scheduling in the physical (SINR) model. Those early results analyzed special settings using e.g. certain node distributions, traffic patterns, transport layers etc. In reality, however, networks often differ from these specialized models and no algorithms were provided to optimize the capacity. On the other hand, graph-based models yielded algorithms like \cite{parth05,ScheidelerRS08} but such models do not capture the nature of wireless communication well, as demonstrated in \cite{gronkvist01,MaheshwariJD08,Moscibroda2006Protocol}. Seven years ago, Moscibroda and Wattenhofer \cite{MoWa06} started combining the best of both worlds, studying algorithms for scheduling in arbitrary networks. Since then, the problems studied in
this setting have reflected the diversity of the application areas underlying it --  topology control \cite{gao08,stoc_topology11,moscibroda06b}, sensor networks \cite{Moscibroda07},
combined scheduling and routing \cite{chafekar07}, ultra-wideband \cite{HuaL06}, and analog network coding \cite{goussevskaia2008complexity}. 

In spite of this diversity, certain canonical problems have emerged, the study of which has resulted in improvements for other problems as well.
The capacity problem is one such problem. After it was quickly
shown to be NP-complete \cite{gouss2007a},  a constant factor approximation 
algorithm for uniform power was achieved in \cite{GHWW09,HW09}, and eventually
extended to essentially all interesting oblivious power schemes in \cite{SODA11}.
In \cite{KesselheimSoda11,KesselheimESA12}, a constant approximation to the capacity problem for arbitrary powers was obtained. 
The relation between capacity using oblivious power and capacity using arbitrary
power was first studied in \cite{us:esa09full}.

Linear power has turned out to be the easiest among fixed power assignments, 
being the only one with a constant factor approximation for
scheduling \cite{FKV09,Tigran11} and a constant-bounded interference measure \cite{FKV09}. Whereas there are instances for
which linear and uniform power are arbitrarily bad in comparison with
mean power \cite{MoWa06}, a maximum feasible subset under mean power
is known to be always within a constant factor of subsets feasible under linear or
uniform power \cite{Tigran11a}.
%
Recently it was shown in \cite{dams2012scheduling} that algorithms for
capacity-maximization in the SINR model can be transferred to a model
that takes Rayleigh-fading into account, losing only an $\BO(\log^* n)$ factor in the approximation ratio. 
This overview is far from being complete, surveys can be found in e.g.  \cite{GoussevskaiaPW10}.

Technically, the idea of looking at the interaction between a feasible
set and a link was studied before. The works of Halld\'orsson \cite{us:esa09full}
and Kesselheim and V\"ocking \cite{KV10} are particularly relevant -- the
first in the context of oblivious-arbitrary comparison, and the second
in the context of oblivious power. Our results improve the bounds in
those papers to the best possible up to a constant factor.





\subsection{Outline of the Paper}
Section \ref{sec:model} lays down the basic setting, including a
formal description of the SINR model. In Section \ref{sec:structural},
we introduce the interference measure and our two structural results. 
We follow this in Section \ref{sec:applications} by illustrating two
applications of these results, one for each of the main theorems.
Section \ref{sec:proofs} contains the proofs of the structural results, 
and Section \ref{sec:fapplications} contains a medley of further applications.

\section{Model and Definitions}
\label{sec:model}

Given is a set $L = \{l_1, l_2, \ldots, l_n\}$ of links, where
each link $l_v$ represents a unit-size communication request from a transmitter
$s_v$ to a receiver $r_v$, both of which are points in an arbitrary metric space.
The distance between two points $x$ and $y$ is denoted $d(x,y)$.
We write $d_{vw} = d(s_v, r_w)$ for short, and denote by $\ell_v$ the length of link $l_v$.
Let $\Delta = \Delta(L)$ denote the ratio between the maximum and
minimum length of a link in $L$.


Let $P_v$ denote the power assigned to link $l_v$, or, in other words, $s_v$ transmits with power $P_v$. In the \emph{physical model} (or \emph{SINR model}) of interference, a
transmission on link $l_v$ is successful if and only if
\begin{equation}
 \frac{P_v/\ell_v^\alpha}{\sum_{l_w \in S \setminus  \{l_v\}}
   P_w/d_{wv}^\alpha + N} \ge \beta\ , 
 \label{eq:sinr}
\end{equation}
where $N$ is a universal constant denoting the ambient noise, $\beta$ denotes the minimum
SINR (signal-to-interference-noise-ratio) required for a message to be successfully received,
$\alpha > 0$ is the so-called path-loss constant,
and $S\subseteq L$ is the set of links scheduled concurrently with $l_v$.

We focus on power assignments $\powp$, where $P_v = \ell_v^{p \cdot \alpha}$.  
This includes all the specific assignments of major interest:
uniform ($\calP_0$), mean ($\calP_{1/2}$), and linear power ($\calP_1$).

%
We say that $S$ is \emph{$\calP$-feasible}, if Eqn.\ \ref{eq:sinr} is
satisfied for each link in $S$ when using power assignment $\calP$.
We say that $S$ is power control feasible (\emph{PC-feasible} for short) if there exists a power
assignment $\calP$ for which $S$ is $\calP$-feasible.
We frequently write simply \emph{feasible} when we refer to PC-feasible.

Let \prob{PC-Capacity} denote the problem of finding a maximum
cardinality subset of the links in $L$ that is PC-feasible (that is we maximize the capacity of the channel used).
Let $OPT^{\calP}(L)$ denote the optimal capacity (i.e., size of the largest $\calP$-feasible subset) of a linkset $L$ under
power assignment $\calP$, and $\PCopt(L)$ denote the optimal capacity under
any arbitrary power assignment (i.e., size of the largest PC-feasible subset). 

\header{Affectance}
We use the notion of \emph{affectance}, introduced in
\cite{GHWW09} and refined in \cite{HW09} and \cite{KV10}.  
The affectance $a^\calP_w(v)$ of link $l_v$ caused by another link $l_w$,
with a given power assignment $\calP$,
is the interference of $l_w$ on $l_v$ relative to the power
received, or
\[ a^\calP_{w}(v) := 
  \min\left(1,c_v \frac{P_w/d_{wv}^\alpha}{P_v/\ell_v^\alpha}\right)
   = \min\left(1, c_v \frac{P_w}{P_v} \cdot \left(\frac{\ell_v}{d_{wv}}\right)^\alpha\right), \]
%
where the factor $c_v := \beta/(1 - \beta N \ell_v^\alpha/P_v)$
depends only on properties of the link $l_v$ and on universal constants. 
We let $a^p_v(w)$ denote $a^{\powp}_v(w)$.
We frequently drop the power assignment reference $\calP$, which means then that we assume $\calP_p$.
Conventionally,  we define $a_v(v) := 0$, since $v$ does not interfere with itself.
For sets $S$ and $T$ of links and a link $l_v$, 
let $a_v(S) := \sum_{w \in S} a_v(w)$, $a_S(v) := \sum_{w \in S} a_w(v)$, and $a_S(T) := \sum_{w\in S} a_w(T)$.
Using this notation, Eqn.\ \ref{eq:sinr} can be rewritten as $a^\calP_S(v) \leq 1$
(except for the near-trivial case of $S$ containing only two links).


We introduce two more affectance notations.
Let $b_v(w) := b_w(v) := a_v(w) + a_w(v)$ be the \emph{symmetric} version of affectance.
Let $\hat{a}_v(w)$ (and $\hat{b}_v(w)$) be the \emph{length-ordered} version,
defined to be $a_v(w)$ (and $b_v(w)$) if $\ell_v \leq \ell_w$ and 0
otherwise, respectively. (This assumes that link-lengths form a total order.)
These are extended in similar ways to affectances to and from sets as
defined for $a_v(w)$. Notice that $a_S(S) = \hat{b}_S(S) = b_S(S)/2$.

\header{(Non)-weak links}
A link is said to be \emph{non-weak} if $c_v \le 2\beta$.
This is equivalent to $\frac{P_v}{\ell_v^{\alpha}} \ge 2 \beta N$.
Intuitively, this means that the link uses at least slightly more power than the absolute minimum
needed to overcome ambient noise (the constant $2$ can be replaced with any fixed constant larger than $1$). Our theorems often assume links to 
be non-weak. This reasonable and often-used assumption \cite{DBLP:conf/infocom/AndrewsD09,Dinitz2010,Goussevskaia2008Local,KV10} 
can be achieved, if necessary, by scaling the powers.

\header{Length classes}
A \emph{length class} is any set $R$ of links with $\Delta(R) \leq 2$ (i.e., link 
lengths vary by a factor no more than $2$). Clearly, any link set $L$ can be partitioned
into $\log \Delta(L)$ length classes. We also refer to this as nearly-equilength class.


\header{Independence}
%
%
We refer to links $l_v$ and $l_w$ as
\emph{$q$-independent} if they satisfy
 $d_{vw} \cdot d_{wv} \ge q^2 \cdot \ell_w \ell_v$.
A set of mutually $q$-independent links is said to be \emph{$q$-independent}. An example of $1$-independence is given in Figure \ref{fig:q-ind}.

\begin{figure*}[ht]
	\begin{center}
		\includegraphics[scale=0.5]{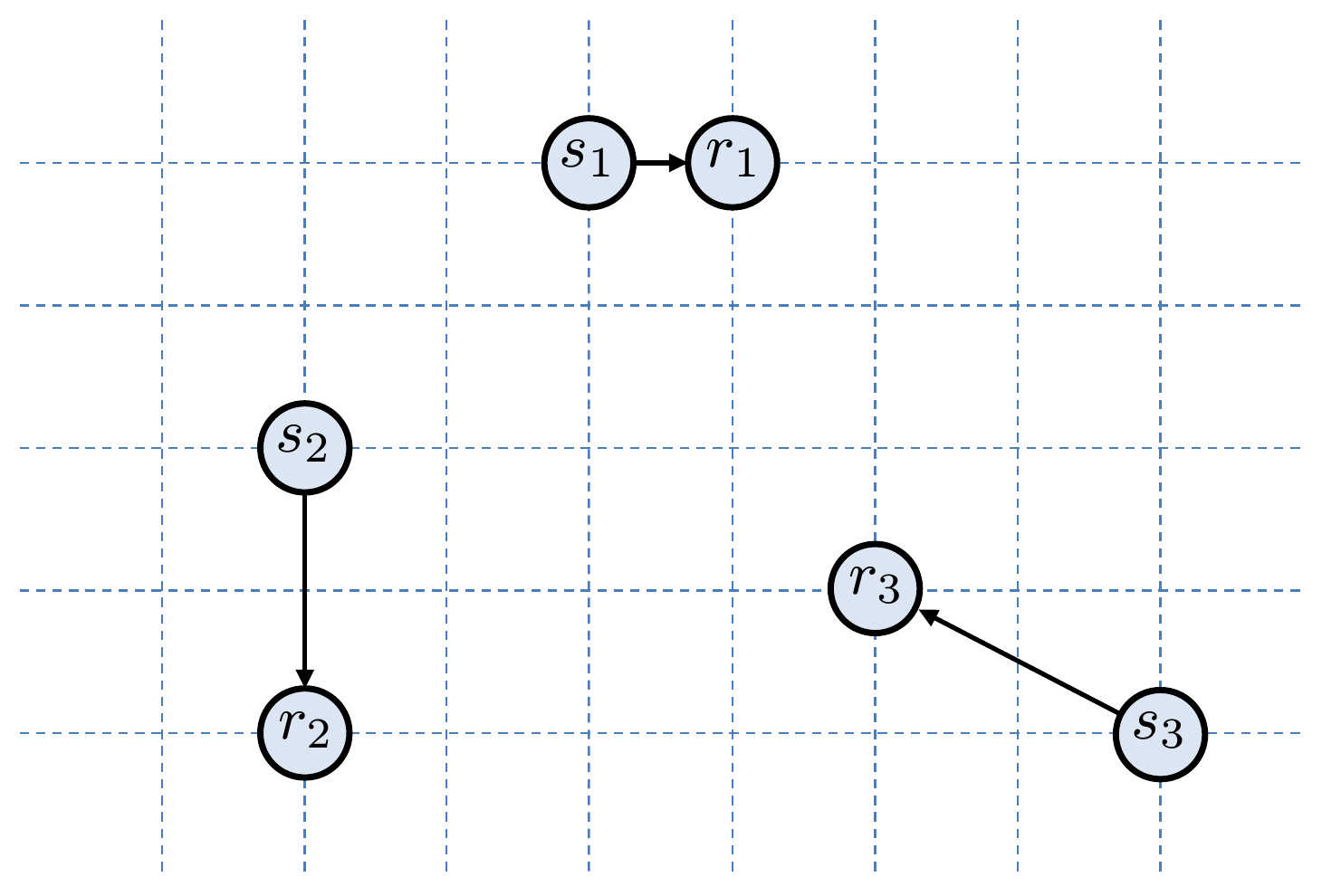}
	\end{center}
	\caption{Links $l_1$ and $l_2$ are $1.92$-independent. Set $\{l_1,l_2,l_3\}$ is $1.43$-independent.}\label{fig:q-ind}
\end{figure*}

Independence is a pairwise property, and thus weaker than feasibility.
The condition is equivalent to $a_v^\calP(u) \cdot a_v^\calP(u) \le \frac{c_v c_w}{q^{2\alpha}}$,
independent of the power assignment $\calP$.
A feasible set is necessarily $\beta^{1/\alpha}$-independent
  \cite{us:esa09full}, but there is no good relationship in the opposite
  direction.

In this paper we provide an independence-strengthening result with better
tradeoffs than the so-called ``signal-strengthening'' result of \cite{HW09}.
The proof is in Appendix \ref{app:model}.

\begin{lemma}
Any feasible set of links can be partitioned into $2q^\alpha/\beta+1$ or fewer $q$-independent sets.
\label{lem:indep}
\end{lemma}

\section{Structural Properties}
\label{sec:structural}

\def\gprop{\mathcal{P}}

We start by defining the interference measure at the center of this work. 

\begin{defn}
Let $L$ be a set of links, $\calP, \calQ$ be two power assignments of $L$, and let $F_{\calQ}(L)$ be the collection of subsets of $L$ that are $\calQ$-feasible. 
Then,
$$I^{\calP}_{\calQ}(L) := \max_{S \in F_{\calQ}(L)} \max_{l_v \in L} \hat{b}^{\calP}_v(S)\ . $$
\end{defn} 
When $\powp$ is used as one (or both) of the assignments, we use $p$ instead of $\powp$ in the sub(super)-scripts -- thus e.g. $I^p_p(L)$ instead of $I^{\powp}_{\powp}(L)$.

As mentioned in the introduction, this definition is analogous to the 
inductive independence number of a graph.
%
In our setting, the weighted graph is formed on the links, that is $L$ is the set of nodes in the graph. The weight 
of the (undirected) edge between
links $l_u$ and $l_v$ is $b_u(v)$ = $b_v(u)$ (computed according to power assignment $\calP$). The ordering is the ascending
order of length. Then, $I^{\calP}_{\calQ}(L)$ is an upper bound on how much
weight/interference (when using power $\calP$) a link can have into a $\calQ$-feasible set containing longer links, just as the 
inductive independence number is an upper bound on how many edges a vertex 
can have to an independent set consisting of higher-ranked vertices.

When using different power assignments $I^{\calP}_{\calQ}(L)$ gives us a handle on comparing the utility of these power assignments. 
We primarily use it in the setting where $\calP = \powp$,
for some $p \in (0, 1]$, and $\calQ$ is (an) optimal arbitrary power assignment (that maximizes the capacity with respect to $L$), 
allowing us to relate oblivious power to arbitrary power.

Here we give two structural results that characterize the utility of
oblivious power assignments. Both of these are best possible and
answer important open questions. The first characterizes the
\emph{price of oblivious power}, i.e., the quality of solutions using
oblivious power assignment relative to those achievable by
unrestricted power assignments. The second is a constant upper bound
on the function when both $\calP$ and $\calQ$ are the same assignment
(specifically, $\powp$ for some $p \in (0, 1]$).

\begin{theorem}
For any set $L$ of non-weak links, any $0 < p < 1$, and any power assignment $\calQ$, $I^{p}_{\calQ}(L) = \BO(\log \log \Delta)$.
\label{thm:effective}
\end{theorem}

\begin{theorem}
Fix a power assignment $\powp$ for any $0 < p \le 1$. 
Then any set $L$ of non-weak links is $\BO(1)$-inductively independent under $\powp$, i.e., $I^p_p(L) = \BO(1)$.
\label{thm:indind}
\end{theorem}
Both theorems will be proven in Section \ref{sec:proofs}.
Theorem \ref{thm:effective} improves upon the $\BO(\log \log \Delta + \log n)$ bound that is stated implicitly in \cite{SODA11} (and extends it to many more power assignments). Theorem \ref{thm:indind} improves
upon the $\BO(\log n)$ bound proven in \cite{KV10}. Both of these new theorems are
optimal (up to constant factors).

\section{Applications}
\label{sec:applications}
Before embarking upon the rather technical proofs of Theorems
\ref{thm:effective} and \ref{thm:indind}, we highlight two
applications, one for each theorem. Further implications are provided
in Section \ref{sec:fapplications}.

\subsection{Capacity Approximation}
\label{sec:capacity}

Using the characterization described above, it is possible to derive a simple single-pass
algorithm for maximizing capacity. This is, in fact, the same
algorithm as used in \cite{SODA11} to maximize fixed power capacity
within a constant factor. 
It is a type of a greedy algorithm that falls under the notion of ``fixed priority'', as defined by 
Borodin et al. \cite{BorodinPriority}.
Recall the $d$-approximation to the max-independent set problem described in the introduction. We added vertices to
the solution set in order, and vertices with edges to the solution set so far were disqualified. Our algorithm below is the natural weighted version of it -- each vertex is assigned a budget of $1/2$, and is disqualified from being in the solution
if the weight of the edges to it from the solution so far exceeds the budget (Lines 4 and 5). We ensure that the final set of links is indeed $\calP_p$-feasible in Line 8.


\begin{algorithm}                      
\caption{Gr(Set $L=\{l_1,l_2,\ldots,l_n\}$ of links in increasing order of length)}          
\label{alg1}                           
\begin{algorithmic}[1]                    
     \STATE $\SetR_0 \gets \emptyset$
     \FOR{$i = 1$ to $n$} 
     \STATE $\SetR_i \gets \SetR_{i-1}$
     \IF{$\hat{b}^p_{\SetR_{i-1}}(l_i) < 1/2$} \label{alg:acceptance}
     	\STATE $\SetR_i \gets \SetR_i \cup \{l_i\}$
     \ENDIF
     \ENDFOR
     \STATE return $X := \{l_v\in \SetR_n : a_{\SetR_n}^p(v)\le 1\}$
\end{algorithmic}
\label{alg1fig}
\end{algorithm}

\begin{theorem}
Let $L$ be a set of links.
For any $\powp$ for which $L$ is non-weak, \alg{Gr} chooses a $\powp$-feasible set $X$ 
such that 
$|X| \ge \frac{|S|}{2\cdot( 2 I^p_{\calQ}(L) + 1)}$ for any power assignment $\calQ$ and any set $S \in F_{\calQ}(L)$.
\label{thm:appx}
\end{theorem}

\begin{proof}
The structure of the proof is inspired by that of, e.g., \cite{KesselheimSoda11}. 
Let $\SetR := \SetR_n$ and $X$ be the sets computed by Algorithm \alg{Gr} on input $L$. 
First, we show that the size of $S$ is not much larger than the size of $\SetR$, and then relate the size of $X$ to $\SetR$ to conclude the statement.   

Consider any power assignment $\calQ$ and feasible set $S$ as specified by the statement of the theorem.
Let $S'$ be $S':=S\setminus \SetR$.
By definition of $I^p_{\calQ}(L)$, we know that  $\hat{b}^p_i(S) \leq I^p_{\calQ}(L)$, for each $l_i \in \SetR$. Thus, 
\begin{equation}
 \hat{b}^p_\SetR(S) \le  I^p_{\calQ}(L) \cdot |\SetR|\ ,
\label{eq:bnd1}
\end{equation}
Now, Algorithm \alg{Gr} chose none of the links in $S'$. Using the acceptance criteria of Line \ref{alg:acceptance}  and the definition of $\hat{b}^p$ yields that 
$\hat{b}^p_\SetR(l_j) \ge \hat{b}^p_{\SetR_{j-1}}(l_j)\ge 1/2$, for each $l_j \in S'$, implying that
\begin{equation}
 \hat{b}^p_\SetR(S') \ge |S'|/2\ . 
\label{eq:bnd2}
\end{equation}
Combining Eqn.\ \ref{eq:bnd1} and Eqn.\ \ref{eq:bnd2},
\[ 
  |S'| \le 2 \cdot \hat{b}^p_\SetR(S') \le 2 \cdot \hat{b}^p_\SetR(S) \le 2 I^p_{\calQ}(L) \cdot |\SetR|\ . \]
Thus,
\begin{equation}
 |S| \le |S'| + |\SetR| \le (2 I^p_{\calQ}(L) + 1) |\SetR|\ . 
\label{eq:bnd3}
\end{equation}
Also, the definition of \alg{Gr} ensures that the average affectance
of links in $\SetR$ is small (at most half). To see this, observe that
the sum of in-affectances is bounded by
\begin{align*}
\sum_{l_i \in \SetR} a_\SetR(l_i) 
&= \sum_{l_i \in \SetR} \sum_{l_j \in \SetR} a_j(i) 
\\
&\overset{1}{=} \sum_{l_i \in \SetR} \sum_{l_j \in \SetR: j < i} (a_j(i)+a_i(j)) 
\\
&\overset{2}{=} \sum_{l_i \in \SetR} \sum_{l_j \in \SetR: j < i} \hat b_j(i)
\\
&\overset{3}{=} \sum_{l_i \in \SetR} \hat b_{\SetR_{i-1}}(i)
\\
& \overset{4}{\leq} \frac12 |\SetR|\ ,
\end{align*}
with the numbered transformation explained as follows:
\begin{enumerate}
\item By rearrangement. Here $j < i$ refers to the indices of the links as sorted by Algorithm \alg{Gr}. We also use that by the definition of affectance, $\sum_{l_i \in \SetR}a_i(i)=0$.
\item From the way \alg{Gr} iterates over the links, $j < i$ implies that $\ell_j \leq \ell_i$. Thus  $\hat b_j(i) =  a_j(i)+a_i(j)$, by definition of $\hat b$.
\item Since $\SetR_{i-1} = \{l_j : l_j \in \SetR, j < i\}$ as specified by \alg{Gr}.
\item By the acceptance criteria of Line \ref{alg:acceptance} of the algorithm.
\end{enumerate}
This implies that the average in-affectance is $\frac1{|\SetR|} a_\SetR(\SetR) \leq \frac12$. 

%
At least half the links have at most double the average
affectance, or
\begin{equation}
|X| = |\{l_v \in \SetR | a_{\SetR}(v) \le 1 \}| \ge \frac{1}{2}|\SetR| \ .
\label{eq:bnd4}
\end{equation}
Combining Eqn.\ \ref{eq:bnd3} and Eqn.\ \ref{eq:bnd4} yields the statement of the theorem.
\end{proof}


\begin{theorem}
For any $\powp$, there is an 
$\BO(\log\log \Delta)$-approximation algorithm for 
\prob{PC-Capacity} that uses $\powp$.
\label{thm:appxpc}
\end{theorem}
%
\begin{proof}
%
%
By Thm.~\ref{thm:appx}, {\alg{Gr}} uses $\powp$ in producing a solution with capacity at most $\BO(1+I^p_{\calQ}(L))$-factor smaller than the optimum for \prob{PC-Capacity}. By Thm.~\ref{thm:effective} this amounts to a $\BO(\log\log \Delta)$ factor.
\end{proof}
\footnote{SH: Why $\BO(1+I^p_{\calQ}(L))$, not $\BO(I^p_{\calQ}(L))$?
Because the interference measure can be much smaller than 1.}

When there is a maximum power level and most links are weak,
we can still attain the same approximation ratio, as done in \cite{SODA11}, by solving the problem separately for the weak links using maximum power.

\subsection{Distributed Scheduling}

A fundamental problem in wireless algorithms is to schedule a given set of links in a minimum number of slots. 
For $\powp$ $(0 \le p \le 1)$, $\BO(\log n)$-approximate centralized algorithms are known \cite{SODA11}.
In \cite{KV10}, the first \emph{distributed} algorithm was given, with an $\BO(\log^2 n)$-approximation ratio. 
Since it is distributed, the algorithm
includes an acknowledgment mechanism (via packets sent from receivers to transmitters) to enable links to know when they have succeeded (and subsequently stop running the algorithm). Assuming ``free'' acknowledgments, \cite{HM11a} improved
the bound to $\BO(\log n)$  (using the same algorithm), but \cite{KV10} remained the best result when acknowledgments have to be implemented explicitly. 

Here we show that, 
\begin{theorem}
There is a randomized distributed $\BO(\log n)$-approximate algorithm 
for \prob{$\powp$-Scheduling} which implements explicit acknowledgments, for any $0 \le p \le 1$. 
\end{theorem}

The cases of $p = 0$ and $p = 1$ was shown in \cite{HM11a}; thus, we only need to focus on $p \in (0, 1)$.
To explain this result, we introduce another complexity measure. 

\begin{defn} \cite{KV10} 
The \emph{maximum average affectance} $A^p(L)$ of a link set 
$L$ is 
$A^p(L) := \max_{S \subseteq L}\frac{a^p_S(S)}{|S|}$.
\end{defn}
It is easily verified that $A^p(L) = \BO(\max_{\calQ} I^p_{\calQ}(L) \cdot \overline{\chi(L)})$,
where $\overline{\chi(L)}$ denotes the minimum number of slots in a feasible schedule of $L$ (using arbitrary power).
Similarly $A^p(L) = \BO(I^p_{p}(L) \cdot \chi^p(L))$ where ${\chi^p(L)}$ denotes the minimum number of slots in a $\powp$-feasible schedule of $L$.

\begin{corollary}
For any set $L$ of links, $A^p(L) = \BO(\log\log \Delta \cdot \overline{\chi(L)})$ and $A^p(L) = \BO({\chi^p(L)})$.
\label{cor:avgaff}
\end{corollary}

It was shown in \cite{KV10} that the distributed scheduling algorithm completes in $\BO(A^p(L) \log n)$ rounds. Thus, the second bound in Corollary \ref{cor:avgaff}
immediately gives us the $\BO(\log n)$-approximation. The approximation bound in \cite{KV10} was worse because it only showed that $A^p(L) = \BO(\chi^p(L) \log n)$.

For comparison with arbitrary power, we can
similarly use Corollary \ref{cor:avgaff} to achieve an $\BO(\log n \cdot 
\log\log\Delta)$-approximation including acknowledgments,
improving on the $\BO(\log n \cdot (\log n + \log\log\Delta))$-factor 
implied by \cite{KV10} and \cite{SODA11}.
Let $\sched$ be the power-control version of the problem.

\begin{corollary}
There is a randomized distributed algorithm for $\sched$ that is
$\BO(\log\log \Delta \cdot \log n)$-approximate with respect to
arbitrary power control optima. It can use any $\powp$ power assignment, $0 < p
< 1$.
\end{corollary}

\section{Proofs of the Structural Results}
\label{sec:proofs}
\subsection{Proof of Theorem \ref{thm:effective}}
We need two lemmas (Lemma \ref{lem:lld} and \ref{lem:affequi})
to bound affectances of a link to and from a set of links.
Denote $\hat{p} := \frac{1}{\min(1-p,p)}$ for the rest of this section.

The first lemma handles the set of long links that have relatively high affectance.
It originates in \cite{us:esa09full} (Lemma 4.4), but is 
generalized here in two ways: To any power assignment $\powp$, and to
sets with the weaker property of $2$-independence. The proof 
is given in Appendix \ref{app:proof-effective}.


\begin{lemma}
Let $p$ be a constant, $0 < p < 1$, $\tau$ be a parameter, $\tau \ge
1$, and $\Lambda = (4 (2\beta \tau)^{1/\alpha})^{\hat{p}}$.  
Let $l_v$ be a link and 
let $Q$ be a 2-independent set of non-weak links in an arbitrary
metric space, where each link $l_w \in Q$ satisfies $\max(a_{v}^\calP(w),a_{w}^\calP(v)) \ge 1/\tau$
and $\ell_w \ge \Lambda \cdot \ell_v$.
Then, $|Q| = \BO(\log\log \Delta)$.
\label{lem:lld}
\end{lemma}

Lemma \ref{lem:lld} bounds the number of longer links that affect a
given link by a significant amount. For affectances below that
threshold, we bound their contributions for each length class separately.

We first need the following geometric argument.
Intuitively, we want to convert statements involving the link $l_v$ 
into statements about appropriate links within the 2-independent set $S$.

\begin{proposition}
Let $l_v$ be a link.  Let $S$ be a 2-independent set of
nearly-equilength links and let $l_u$ be the link in $S$ with $d_{uv}$
minimum. Then, $d_{wv}\geq d_{wu}/6$, for any $l_w$ in $S$.
\label{prop:geom1}
\end{proposition}
The reader may find Figure \ref{fig:geom1} helpful when reading the proof of Proposition \ref{prop:geom1}.
\begin{proof}
Let $D = d_{wv}$ and note that by definition $d_{uv} \le D$.
By the triangular inequality and the definition of $l_u$, 
\begin{align}
 d_{wu} &\le d_{wv} + d_{uv} + d_{uu}
 \nonumber\\
 & = d_{wv} + d_{uv} + \ell_u 
       \le 2 D + \ell_u\ . 
\label{eq:dwu}
\end{align}
Similarly, 
\begin{equation}
 d_{uw} \le d_{uv} + d_{wv} + \ell_w \le 2 D + \ell_w \ .
\label{eq:duw}
\end{equation}
Applying 2-independence, on one hand, and multiplying Eqn.\ \ref{eq:dwu} and Eqn.\ \ref{eq:duw}, on the other hand, we have that
\[ 4 \ell_u \ell_w \le d_{wu} \cdot d_{uw} \leq (2D + \ell_u) \cdot (2 D + \ell_w) \ . \]
This implies that $D$ must be at least $\min(\ell_u, \ell_w)/2$ which in turn is at least $\max(\ell_u, \ell_w)/4$, using that the links are nearly-equilength.
Thus we can bound $l_u\leq 4D$ in Eqn.\ \ref{eq:dwu} and obtain $d_{wu}\leq 6D$.
\end{proof}
\begin{figure*}[ht]
	\begin{center}
		\includegraphics[scale=0.5]{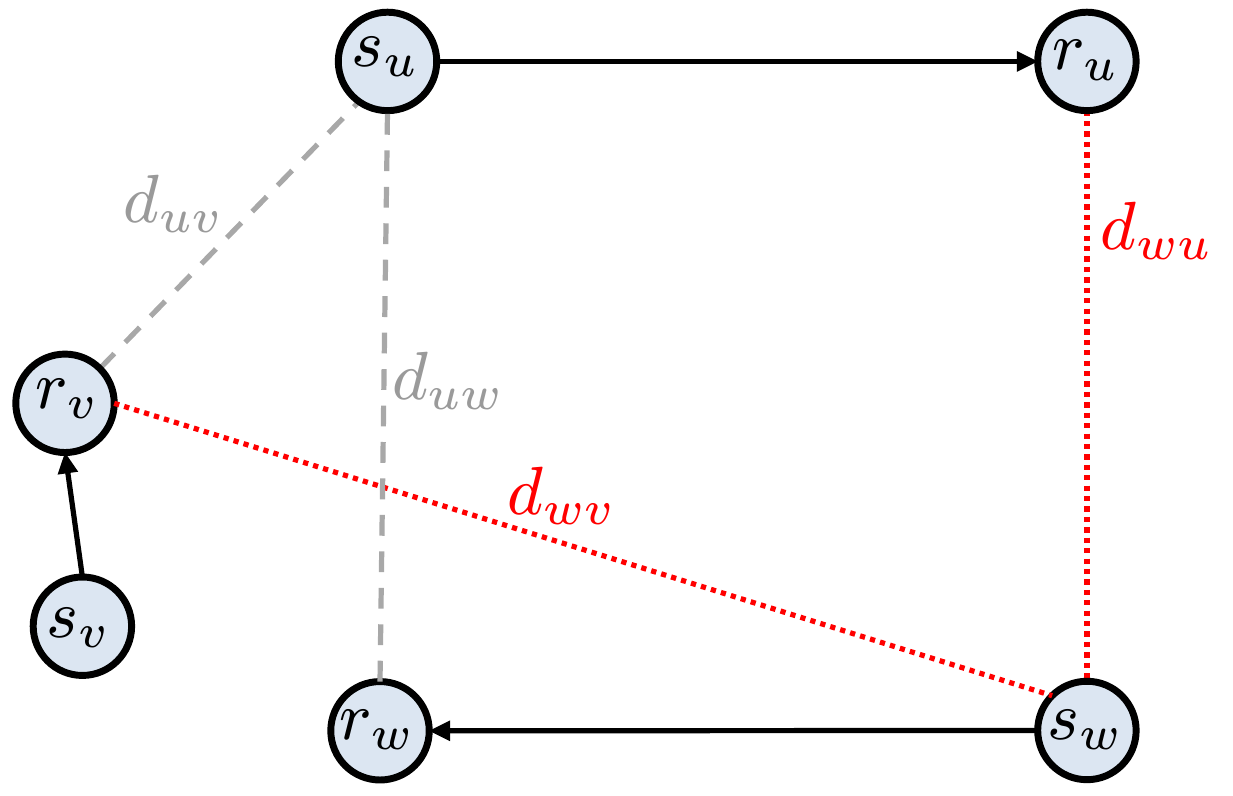}
	\end{center}
	\caption{Displays links $l_u,l_w$ and $l_w$ as used in the proofs of Propositions \ref{prop:geom1} and \ref{prop:geom2}. The distances $d_{wv}$ and $d_{wu}$ that are related to each other in the Proposition's statement are represented by red dotted lines. The gray dashed lines mark distances $d_{uw}$ and $d_{uv}$ that are used in the proofs as well.}\label{fig:geom1}
\end{figure*}

\begin{proposition}
Let $l_v$ be a link.  Let $S$ be a 2-independent set of
nearly-equilength links and let $l_u$ be the link in $S$ with $d_{uv}$
minimum. Then, $d_{vw}\geq d_{wu}/6$, for any $l_w$ in $S$.
\label{prop:geom2}
\end{proposition}

\begin{proof}
The proof is essentially the same as the proof of Proposition \ref{prop:geom2} with the roles of senders and receivers switched. For completeness, the proof can be found in Appendix \ref{app:proof-effective}.
\end{proof}

This leads to the second lemma of this section.
\begin{lemma}
Let $q$ be a positive real value and $l_v$ be a link.
Let $S$ be a $2$-independent and feasible set of
non-weak links belonging to a single length-class of minimum length at least
$q^{\hat{p}/\alpha} \cdot \ell_v$.
Then, 
\[b^p_v(S) \le \max_{\{l_w,l_{w'}\} \subseteq S} b^p_v(\{w,w'\}) + \BO(1/q).\]
\label{lem:affequi}
\end{lemma}
\begin{proof}
Consider the link $l_u$ in $S$ with $d_{uv}$ minimum.
Since $\ell_v \le \ell_u$, it holds that $c_v \le c_u$.
Then, we have that
\[ a_w^p(v) = c_v \left(\frac{\ell_v^{1-p} \ell_w^p}{d_{wv}}\right)^\alpha
\le 
c_u \left(\frac{(\ell_u/q^{\hat{p}/\alpha})^{1-p} \ell_w^p}{d_{wu}/6}\right)^\alpha\ ,\]
using Proposition \ref{prop:geom1}. Continuing from above,
The above equals
\[ a_w^p(v) \le
 \frac{6^{\alpha}}{q^{\hat{p} \cdot (1-p)}} \cdot c_u\frac{\ell_w^{p\alpha}}{\ell_u^{p\alpha}}\cdot \left( \frac{\ell_u}{d_{wu}}\right)^\alpha
  =
 \frac{6^{\alpha}}{q^{\hat{p} \cdot (1-p)}} a_w^p(u)
 \le
  \frac{6^{\alpha}}{q} a_w^p(u)\ , \]
where the last inequality follows from our choice of $\hat{p}=\frac{1}{\min(1-p,p)}$.

For any subset $S'\subseteq S\setminus \{l_u\}$, this extends to 
\begin{equation}
a_{S'}^p(v)\leq \frac{6^{\alpha}}{q}  a_{S'}^p(u)=\BO(1/q),\label{bd:aff-bound1}
\end{equation}
as $S'$ is feasible.

Now consider the link $l_{u'}$ in $S$ with $d_{vu'}$ minimum. 
Since links in $S$ are non-weak, $c_w \le 2 c_{u'}$.
Thus,
\[ a_v^p(w) = c_w \left(\frac{\ell_v^p \ell_w^{1-p}}{d_{vw}}\right)^\alpha
 \le  2 c_{u'} \left(\frac{(\ell_w/q^{\hat{p}/\alpha})^p \ell_w^{1-p}}{d_{wu'}/6}\right)^\alpha, \]
using Proposition \ref{prop:geom2} and the assumed bound on link lengths.
 Since $l_{u'}$ and $l_w$ are nearly-equilength, this is bounded by
\[ a_v^p(w) \le  2 c_{u'} \left(\frac{(\ell_w/q^{\hat{p}/\alpha})^p (2\ell_{u'})^{1-p}}{d_{wu'}/6}\right)^\alpha, \]
Rearranging, we get that
\[ a_v^p(w) \le  2^{1+\alpha}\cdot \frac{6^{\alpha}}{q^{\hat{p} \cdot p}} 
  \cdot c_{u'}\frac{\ell_w^{p\alpha}}{\ell_{u'}^{p\alpha}}\cdot \left( \frac{\ell_{u'}}{d_{wu'}}\right)^\alpha
  \le  2^{1+\alpha} \cdot \frac{6^{\alpha}}{q}\cdot a_w^p\left(u'\right)\ . 
\]
For any subset $S'\subseteq S\setminus \{l_{u'}\}$ this extends to 
\begin{equation}
a_{v}^p\left(S'\right)\leq  2^{1+\alpha} \cdot\frac{6^{\alpha}}{q}  a_{S'}^p\left(u'\right)=\BO(1/q),\label{bd:aff-bound2}
\end{equation}
since $S'$ is feasible. 
Combining (\ref{bd:aff-bound1}) and (\ref{bd:aff-bound2}) yields
\[ b^p_v(S) - b^p_v\left(\left\{u,u'\right\}\right) 
= a_{S\setminus \{u,u'\}}^p(v) + a_v^p\left(S\setminus \left\{u,u'\right\}\right) = \BO(1/q), \]
from which we conclude that 
\[b^p_v(S) \le \max_{\{l_w,l_{w'}\} \subseteq S} b^p_v\left(\left\{w,w'\right\}\right) + \BO(1/q).\]
\end{proof}



We are now ready to prove the core result, Theorem \ref{thm:effective}.

\smallskip


\begin{proof}{[of Theorem \ref{thm:effective}]}
Choose any $l_v \in L$ and any $\calQ$-feasible subset $S \subseteq L$. We shall show that $\hat{b}^p_v(S)  = \BO(\log \log \Delta)$. By the definition of $\hat b$,
we can assume that all links in $S$ are larger than $l_v$, since $\hat b$ is defined in such a way that all shorter links do not contribute to its value. With this assumption, $\hat b^p(S) = b^p(S)$. 
We use the independence-strengthening lemma (Lemma \ref{lem:indep}) 
to partition $S$ into
at most $\frac{2^{1+\alpha}}{\beta}+1$ different 2-independent feasible sets. 
Let $S'$ be one of these sets.

Let $D := \log \Delta(L)$ and let $\Lambda = (4 (4\beta D)^{1/\alpha})^{\hat{p}}$.
We say that a link $l_w$ in $S$ is \emph{short} if
$\ell_v \le \ell_w < \Lambda \cdot \ell_v$
and \emph{long} if $\ell_w \ge \Lambda \cdot \ell_v$.
We partition $S'$ into three sets:

\begin{description}
\item[$S_1$] Long links $l_w$ with $b_v(w) \ge 1/D$,

\item[$S_2$] Long links $l_w$ with $b_v(w) < 1/D$, and

\item[$S_3$] Short links.
\end{description}
We bound the affectance $b_v(S_i)$ of each set $S_i$ separately. 

As $S_1\subseteq S$ is $2$-independent, since its superset $S$ is $2$-independent, we can apply Lemma \ref{lem:lld} with $\tau=2D$. This implies that $|S_1| = \BO(\log\log \Delta(S_1))$ and thus 
\[b_v(S_1) \le 2|S_1| = \BO(\log\log \Delta(S)) =
\BO(\log\log \Delta(L)).\]

Next we observe that due to the choice of $D$, the set $S$ (and thus $S_2$) can be partitioned into $\BO(D)$ length classes $X_1, X_2, \dots$. 
Each such class $X_i$
satisfies the hypothesis of Lemma \ref{lem:affequi} with $q:=D \geq 1$.
Since $b_v(w) < 1/D$, for each $l_w \in X_i$, by assumption, 
Lemma \ref{lem:affequi} gives that
\[b_v(X_i) = \BO(1/D)\text{ and }b_v(S_2) = \BO(D) \cdot b_v(X_i) = \BO(1).\] 

The set $S_3$ can be partitioned into $\log \Lambda = \BO(D)$ 
length classes $Y_1,\dots,Y_{\log \Lambda}$ as well. For each such length class $Y_i$,
we apply Lemma \ref{lem:affequi} with $q = 1$, giving that
$b_v(Y_i) = \BO(1)$, for a total of $b_v(S_3) = \BO(D) = \BO(\log\log \Delta)$. Thus, 
\[b_v(S') = b_v(S_1) + b_v(S_2) + b_v(S_3) = \BO(\log\log \Delta),\]
and
\[b_v(S) \le \left(\frac{2^\alpha}{\beta}+1\right) b_v(S') = \BO(\log\log \Delta).\]  
\end{proof}
To provide more intuition behind this proof, consider Figure \ref{fig:intuition}, where $\Delta$ is a small constant. Therefore\todo{MMH: ??} the set $S_3$ contains only links of similar length to $l_v$. Since $S_3$ is feasible for some power assignment, these links are not too close to each other and can thus be scheduled within a few time slots using $\powp$. Since all long links in this example are of roughly the same length, they can be partitioned in $S_1$ and $S_2$ using a disc around $l_v$. Note that in general, the radius of the disc that is used to decide whether a link is in $S_1$ or $S_2$ depends on the link's length relative to $l_v$. Now we see that $l_v$ hardly interferes with a set of feasible links far away from $l_v$. There can also not be too much interference with long links that are close to $l_v$ since there can't be too many of them within the disc as they also need to be feasible for some power assignment.
\begin{figure*}[ht]
	\begin{center}
		\includegraphics[scale=0.5]{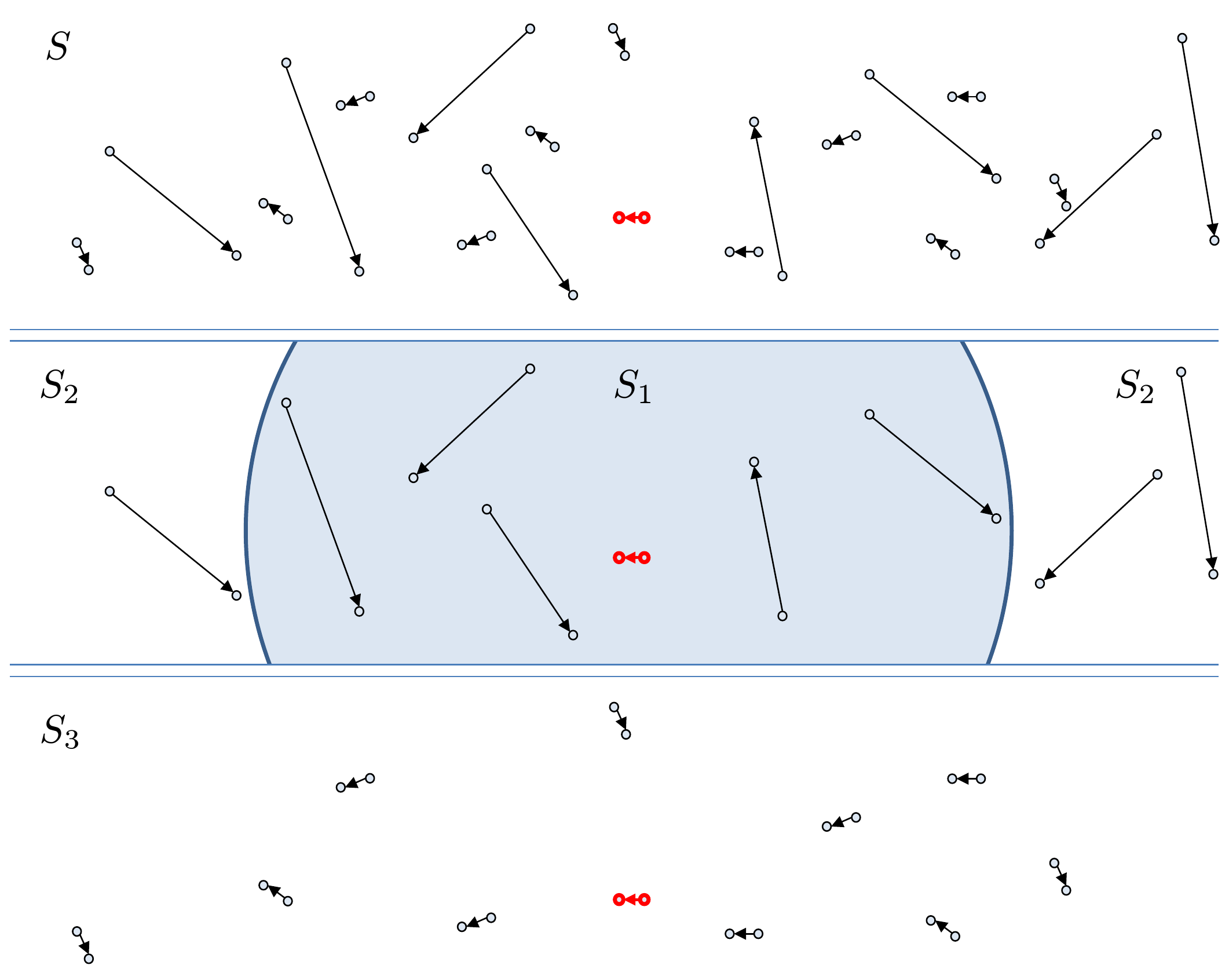}
	\end{center}
	\caption{The top part of the figure displays a set $S$ of links that is (assumed to be) feasible under some power assignment. Let $l_v$ in the proof of Theorem \ref{thm:effective} be the red bold link. The center part displays long links of $S$ partitioned (in this particular example) into $S_1$ and $S_2$ using a disc around $l_v$, the lower part displays $S_3$ consisting of short links.}\label{fig:intuition}
\end{figure*}

%

\subsection{Proof of Theorem \ref{thm:indind}}

The following lemma is the crucial element in the proof.

\begin{lemma}
Let $L$ be a $\powp$-feasible set of non-weak links and $l_v$ be a link (not necessarily in $L$). 
Then, $\hat{a}_v(L) = \BO(1)$.
\label{lem:constoutaff}
\end{lemma}
\begin{proof}
Let $\calL(n)$ be the set of all $\powp$-feasible sets of non-weak links of size $n$.
Define $g(n)$ (a function of $n$) to be the ``optimum upper bound'' on $\hat{a}$, that is, $g(n) := \sup_{L \in \calL(n)} \sup_{l_v} \hat{a}_v(L)$. Such a function exists, since $\hat{a}_v(L) \leq n$ for any set $L$ of size $n$ and any $l_v$. 
We claim that $g(n)$ is indeed $\BO(1)$, which implies the lemma. 
For contradiction, assume $g(n) = \omega(1)$. 

Since $g(n) = \omega(1)$, we can choose a large enough $n_0$ such that all of the following hold:
\begin{enumerate}
\renewcommand{\labelenumi}{(\alph{enumi})}

\item There exists $L \in \calL(n_0)$ and $l_v$ such that:
\begin{equation}
\hat{a}_v(L) \geq \frac{1}{2} g(n_0) \ . \label{AisLarge}
\end{equation}

Observe that such an $L$ and $l_v$ always exist independent of $n_0$, by the definition of $g$.

\item Define $f(n) :=  \frac{1}{2} 2^{\frac{1}{4c_3}g(n)}$, 
where $c_3$ is a fixed constant to be specified later.
Then,
\begin{equation}
\label{fnlarge} 
f(n_0) \ge (\fncons2)^{1/(p\alpha)}\ .
\end{equation}

\item Lastly,
\begin{equation}
\label{gnlarge} 
g(n_0) \geq 16 \cdot (4^{\alpha}+1)
\end{equation}
\end{enumerate}

%

We prove our lemma by deriving a contradiction to Eqn.\ \ref{AisLarge}. 
To prove this, 
we partition the link set $L$ into $L_1$ and $L_2$ where $L_1 := \{l_w : \ell_w \leq f(n_0) \cdot \ell_v\}$ and $L_2 := L \setminus L_1$ .

\begin{claim}
\label{affonl1}
$\hat a_v(L_1) < \frac14 g(n_0)$. 
\end{claim}
\begin{proof}
By definition of $\hat a$, we can ignore links in $L_1$ smaller than $l_v$. Since the maximum length in $L_1$ is at most $f(n_0) \cdot \ell_v$, the remaining links in
 $L_1$ can be divided into $\log f(n_0)$ length classes. 
 Consider any such length class $C$.
By Lemma \ref{lem:indep}, $C$ can be partitioned into $\frac{2^{\alpha+1}}{\beta} +1$ sets that are feasible and $2$-independent. For any such set $C'$, we 
can invoke Lemma \ref{lem:affequi} to show that $a_v(C') = \BO(1)$ and thus 
$a_v(C) = (2^{1+\alpha}/\beta +1) \BO(1) = \BO(1)$. By setting $c_3$ to be this constant, we get that
\[ \hat{a}_v(L_1) \le c_3 \log f(n_0) 
  = c_3 \left(\frac1{4 c_3} g(n_0) - 1\right) < \frac14 g(n_0)\ , \]
where we used the definition of $f(n)$ in the equality.
\end{proof}

\begin{claim}
$\hat a_v(L_2) \le \frac{1}{4}g(n_0)$,
\label{affonl2}
\end{claim}
\begin{proof}
%
Consider $l_w  \in L_2$ such that $D := d(s_v, s_w)$ is
minimized.
Let $L_3$ be the set of links in $L_2$ with 
receivers within the ball $B(s_v, D/2)$ of radius $D/2$ around $s_v$, 
and set $L_4 := L_2 \setminus L_3$.

Let us first handle affectances to $L_3$ using the following  (proof in Appendix \ref{app:proof-effective}):
\begin{proposition}
$|L_3| \leq  2 \cdot 4^{\alpha} + 1$.
\label{prop:l3bound}
\end{proposition}
%
%

Using this proposition, we get that
\[ \hat{a}_v(L_3 \cup \{l_w\}) \leq |L_3| +1\leq   2 \cdot (4^{\alpha} + 1) \le \frac{1}{8} g(n_0)\ , \]
\footnote{SH: Why can we assume $a_v(w)\leq 1$? There is no restriction on $l_v$ nor its location. MMH: Definition of affectance.}
where the last inequality follows from Eqn.\ \ref{gnlarge}. 

Now consider any $l_u \in L_4 \setminus \{l_w\}$.
Using that $r_u$ is at least $D/2$ away from $s_v$ (due to being in $L_4$) and the fact that we chose $D:=d(s_v,s_w)$, the triangle inequality yields $d(s_v, r_u) \geq \frac13 d(s_w, r_u)$. Thus,
\begin{equation}
a_{v}(L_4 \setminus \{\ell_w\}) 
  \le \sum_{\ell_u \in L_4 \setminus \{\ell_w\}} c_u\cdot \frac{P_v}{d(s_v, r_u)^{\alpha}} \frac{\ell_u^{\alpha}}{P_u} 
\le
  3^{\alpha} \sum_{u} \frac{P_v}{P_w} \frac{P_w}{d(s_w, r_u)^{\alpha}} \frac{\ell_u^{\alpha}}{P_u}
= 3^{\alpha} \frac{P_v}{P_w} a_w(L_4)\ .
\label{in:1} 
\end{equation} 

\begin{figure*}[ht]
	\begin{center}
		\includegraphics[scale=0.5]{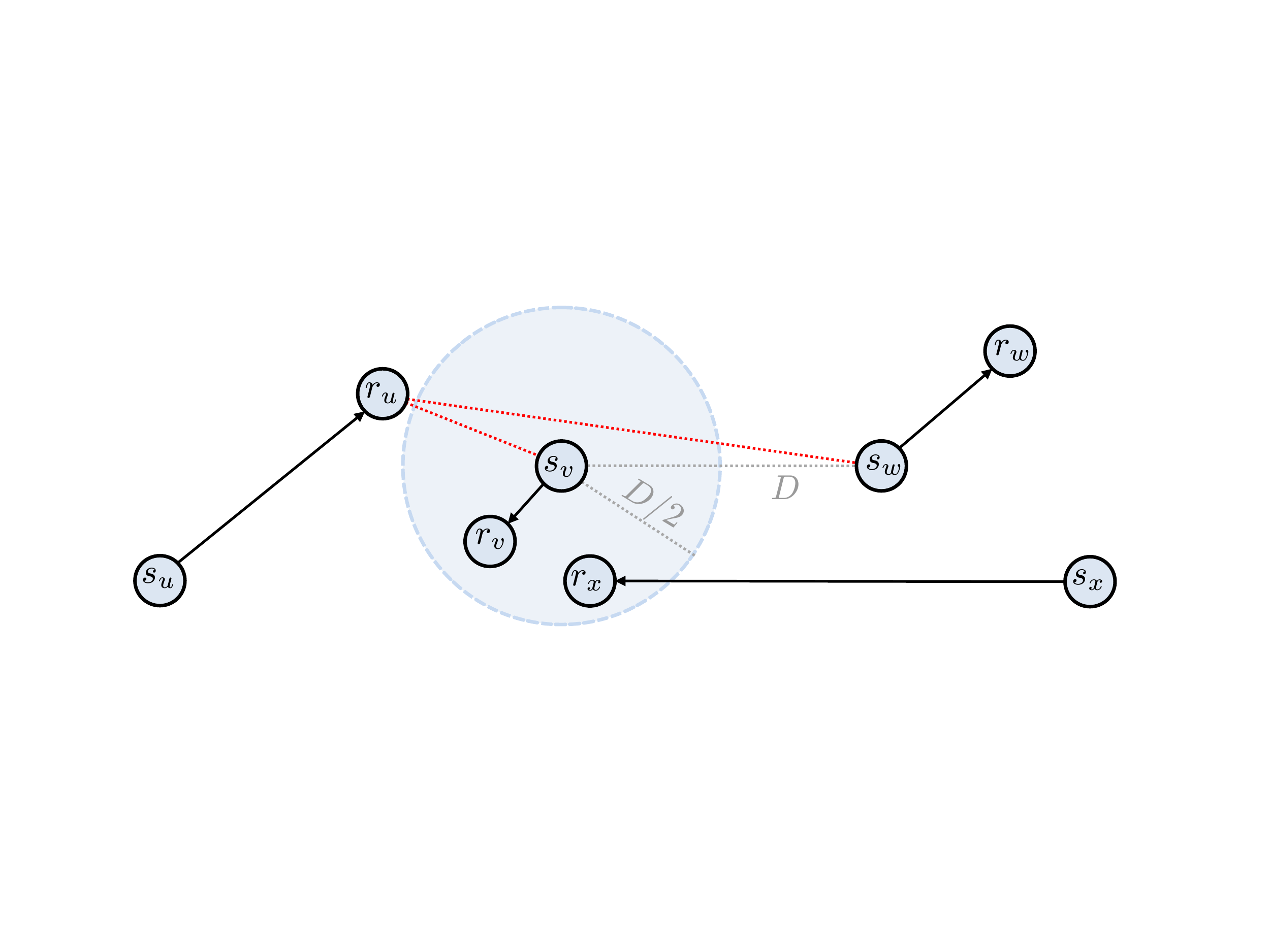}
	\end{center}
	\caption{Nodes $s_v,s_w$ and $r_u$ play the role described in the proof. Here, $L_2:=\{l_u,l_w,l_x\}, L_3:=\{l_x\}$ and $L_4:=\{l_u,l_w\}$. The red dotted lines indicate the relevant distances in the triangle inequality yielding $d(s_v,r_u)\geq \frac{1}{3}d(s_w,r_u)$.}\label{fig:triangle-eq}
\end{figure*}


Since the power function $\powp$ is non-decreasing and $\ell_w\geq f(n_0)\cdot \ell_v$, due to the choice of $L_2\supseteq L_4$, 
\[P_ w \geq \powp(f(n_0) \cdot
\ell_v) = f(n_0)^{p\alpha} P_v.\]
Thus, $\frac{P_v}{P_w} \leq \frac1{f(n_0)^{p\alpha}} \leq \frac{1}{16 \cdot 3^{\alpha}\beta}$ using Eqn.\ \ref{fnlarge}.
Combining this insight with Inequality \ref{in:1} and using that $a_{w}(L_4) \leq g(n_0)$ due to the definition of $g(n)$, we conclude that
\[ \hat{a}_{v}(L_4  \setminus \{l_w\}) 
     \le  3^\alpha \frac1{8 \cdot 3^{\alpha}} g(n_0)
      = \frac{1}{8} g(n_0) \ . \]
\footnote{MMH: What was the meaning of this statement here: ``Therefore the assumption $g(n)=\omega(1)$ is wrong''?}
This completes the proof of Claim \ref{affonl2}.
\end{proof}

Combining Claims \ref{affonl1} and \ref{affonl2},
we get that $\hat{a}_v(L) < \frac{1}{2} g(n_0)$, 
contradicting Eqn.\ \ref{AisLarge}. 
This completes the proof of Lemma \ref{lem:constoutaff}.
\end{proof}

\noindent We can now complete the proof of Theorem \ref{thm:indind},
which we recall states that $I_p^p(L) = O(1)$, for any $p \in (0,1)$ and any set $L$ of links.
\todo{MMH: Restate the theorem?}
\begin{proof}
Consider any $S \in F_{\powp}(L)$ and any $l_v \in L$. 
Starting with the definition of $\hat{b}$,
\[ \hat b_v(S) = \hat{a}_v(S) + \hat{a}_S(v) = O(1) + O(1) = O(1)\ , \]
where we apply Lemma \ref{lem:constoutaff} on the first term and Lemma 7 of \cite{KV10} on the second term.
\end{proof}

We remark that the bound in neither theorem remains true when there are weak links.

\section{Further Applications}
\label{sec:fapplications}

Both of our structural
results have a number of further applications, improving the approximation ratio for many fundamental and important
problems in wireless algorithms. All our improvements come from noticing
that many existing approximation algorithms have bounds that are implicitly 
based on $I^p_{\calQ}(L)$ or $I^p_p(L)$ (or both).
Plugging in our improved bounds for these thus gives the (poly)-logarithmic
improvements for a variety of applications. Here we often omit 
proofs of our claims, as they are all of the same flavor.

\subsubsection*{Connectivity}
Wireless connectivity --- the problem of \emph{efficiently} connecting a set of wireless nodes in an interference aware manner --- is one of the central problems in wireless network research \cite{HM12}. Such a structure may underlie a multi-hop wireless network, or
provide the underlying backbone for synchronized operation of an ad-hoc network. In a wireless sensor
network, the structure can function as an information aggregation mechanism.

Recent results have shown that any set of wireless nodes can be strongly connected
in $\BO(\log n \cdot (\log n + \log\log\Delta ))$ slots using mean power in both centralized \cite{HM12} and distributed \cite{PODC12} algorithms. These results are directly improved by Theorem \ref{thm:appxpc}:

\begin{theorem}
Any set of links can be strongly connected in $\BO(\log n \cdot \log\log\Delta)$ slots using power assignment $\powp$. This can be computed
by either a poly-time centralized algorithm or an $\BO(poly(\log n) \log \Delta)$-time distributed algorithm.
\end{theorem}


Results for variations of connectivity such as \emph{minimum-latency aggregation scheduling} and applications of connectivity such as maximizing the aggregation rate in a sensor network benefit from similar improvements. We refer the reader to 
\cite{HM12} for a discussion of these problems and their numerous applications.

\subsubsection*{Spectrum Sharing Auctions}
In light of recent regulatory changes by the Federal Communications Commission (FCC) opening up the possibility of dynamic white space networks (see, for example, \cite{DBLP:conf/sigcomm/BahlCMMW09}), the problem of dynamic allocation of channels to bidders (these are the wireless devices) via an auction has 
attracted much attention \cite{Zhou:2008:ESS:1409944.1409947,DBLP:conf/infocom/ZhouZ09}. 

The combinatorial auction problem in the SINR model is as follows: Given $k$ identical channels and $n$ users (links), with each user having a valuation for each of the $2^k$ possible
subset of channels, find an allocation of the users to channels so that
each channel is assigned a feasible set and the social welfare is maximized.

For the SINR model, recent work \cite{DBLP:conf/spaa/HoeferKV11,DBLP:conf/sigecom/HoeferK12} has established a number of results depending on different valuation functions. Since
these results are based on the inductive independence number, Theorem \ref{thm:indind} improves virtually all of them by a $\log n$ factor.
For instance, an algorithm was given in \cite{DBLP:conf/spaa/HoeferKV11} for general valuations 
that achieves an $\BO(\sqrt{k} \log n \cdot I^p_p(L)) = \BO(\sqrt{k} \log^2
n)$-approximation. We achieve an improved result by simply plugging in Theorem \ref{thm:indind}.

\begin{corollary}
Consider the combinatorial auction problem in the SINR setting, for
any fixed power assignment $\powp$ with $0 < p \leq 1$.
There exist algorithms that achieve an $\BO(\sqrt{k}\log n)$-factor for
general valuations \cite{DBLP:conf/spaa/HoeferKV11}, a
$\BO(\log n + \log k)$-approximation for symmetric valuations and an $\BO(\log n)$-approximation for Rank-matroid valuations \cite{DBLP:conf/sigecom/HoeferK12}.
\end{corollary}

\subsubsection*{Dynamic Packet Scheduling}


Dynamic packet scheduling to achieve network \emph{stability} is one of the fundamental problems in (wireless) network
queuing theory \cite{TE92}. In spite of its long history, this
fundamental problem has been considered only recently in the SINR model 
(see \cite{lqfmobihoc,kesselheimStability,sirocco12}). The problem calls for an algorithm that can keep queue sizes bounded in a wireless network under stochastic arrivals of packets at transmitters. A measure called \emph{efficiency} between $0$ and $1$ is used to capture how well a given algorithm performs compared to a hypothetical best algorithm. We refer the reader to the aforementioned papers for exact definitions and motivations related to this problem.

The state-of-the-art results for this problem have been achieved very recently and simultaneously in \cite{sirocco12} and \cite{kesselheimStability}. In spite of differences in the algorithm  and assumptions made, both are based on the scheduling
algorithm of \cite{KV10} and achieve a similar result. Recall that the maximum average affectance is $A^p(L) = \max_{S\subseteq L}\frac{a^p_S(S)}{|S|}$ and $\chi^p(L)$ is the minimum number of slots in a $\powp$-feasible schedule of $L$. Let $\phi(L) = \frac{A^p(L)}{\chi^p(L)}$. 


The result in \cite{kesselheimStability,sirocco12} can be succinctly expressed as follows.

\begin{theorem}{\cite{kesselheimStability,sirocco12}}
There exists a distributed algorithm that achieves $\Omega\left(\frac1{\log n \cdot (1+ \phi(L))}\right)$-efficiency for any link set $L$. 
\end{theorem}
Since the best bound on $\phi(L)$ known was $\BO(\log n)$ \cite{KV10}, both papers claimed $\Omega(\frac1{\log^2 n})$-efficiency. Results in this paper show that $\phi(L) = \BO(1)$ (see second part of Corollary \ref{cor:avgaff}), which gives the following improved result:

\begin{corollary}
There exists a distributed algorithm that achieves $\Omega\left(\frac1{\log n}\right)$-efficiency for any power assignment $\powp$ ($0 < p \le 1$). 
\end{corollary}

%

Since Corollary \ref{cor:avgaff} also shows that $\overline{\phi(L)} = \frac{A^p(L)}{\overline{\chi(L)}} = \BO(\log n \cdot \log\log\Delta)$, we also get the following improved bound
for power control:
\begin{corollary}
There is a distributed algorithm with $\Omega\left(\frac{1}{\log n \cdot \log\log \Delta}\right)$-efficiency, with respect to power control optima.
\end{corollary}

\textbf{Acknowledgments:} We would like to thank Marijke Bodlaender for helpful comments and for pointing out errors in the conference version \cite{us:SODA13}.
 



\bibliographystyle{abbrv}
\bibliography{references}		

\appendix

\section{Missing Proof from Section \ref{sec:model}: Independence Strengthening}
\label{app:model}

\noindent \textbf{Lemma \ref{lem:indep}}\ 
\emph{
Any feasible set of links can be partitioned into 
$\lfloor \frac{2q^\alpha}{\beta}\rfloor +1$ or fewer $q$-independent sets.
}
\begin{proof}
Let $S$ be a feasible set and 
$\calP$ a power assignment such that $S$ is feasible for $\calP$.  We form
a graph $G$ on linkset $S$, such that two links $l_v$ and $l_w$ are
adjacent if $b^\calP_v(w) \ge \beta/q^\alpha$.
Let $Z$ be $Z := \lfloor 2q^\alpha/\beta \rfloor$.

 
We first show that $G$ is $Z$-inductive (a.k.a.\ $Z$-degenerate,
or Szekeres-Wilf number $Z$), which means that there is an ordering of the
vertices so that each vertex has at most $Z$ neighbors that appear
later in the ordering.

Since $S$ is feasible, $a_S^\calP(v) \le 1$, for any $l_v$ in $S$.
Thus, 
\[b^\calP_S(S)/2 = a_S^\calP(S) \le |S|,\]
 implying that some link $l_u$ satisfies
\[ b^\calP_u(w) \le 2 \ . \]
It is then clear that for at most $Z = \lfloor 2q^\alpha/\beta \rfloor$ links 
$l_w$ it is true that $b^\calP_u(w) \ge \beta/q^\alpha$.
We then form a $Z$-inductive ordering of $S$
by placing $l_u$ first, followed by the inductively constructed ordering for $S
\setminus \{l_u\}$. 

Since $G$ is $Z$-inductive, it is $(Z+1)$-colorable.
Consider a color class (a stable set) $C$.
It holds by definition for any pair $l_v, l_w$ of links in $C$
that
\[ a_w^\calP(v) \cdot a_v^\calP(w) 
   \le \frac{\beta}{q^{\alpha}} \cdot \frac{\beta}{q^{\alpha}} 
   \le \frac{c_v c_w }{q^{2 \alpha}}\ , \]
which implies that $l_v$ and $l_w$ are $q$-independent.
Quantifying over all pairs in $C$, it follows that $C$ is
$q$-independent. 
\end{proof}

\section{Missing proofs from Section \ref{sec:proofs}}
\label{app:proof-effective}




\noindent \textbf{Proposition \ref{prop:geom2}}  Let $l_v$ be a link.  Let $S$ be a 2-independent set of
nearly-equilength links and $l_u$ be the link in $S$ with $d_{vu}$
minimum. Then $d_{vw}\geq d_{wu}/6$.
\\

The reader might find it useful to use Figure \ref{fig:geom2} while reading the proof of Proposition \ref{prop:geom2}.

\begin{proof}
Let $D = d_{vw}$ and note that by definition $d_{vu} \le D$.
By the triangular inequality and the definition of $l_u$, 
\begin{align}
 d_{uw} &\le d_{uu} + d_{vu} + d_{vw}
 \nonumber\\
 & = \ell_u + d_{vu} + d_{vw} 
       \le 2 D + \ell_u\ . 
\label{eq:duw2}
\end{align}
Similarly, 
\begin{equation}
 d_{wu} \le d_{vu} + d_{vw} + \ell_w \le 2 D + \ell_w \ .
\label{eq:dwu2}
\end{equation}
Applying 2-independence, on one hand, and multiplying Eqn.\ \ref{eq:dwu2} and Eqn.\ \ref{eq:duw2}, on the other hand, we have that
\[ 4 \ell_u \ell_w \le d_{uw} \cdot d_{wu} \leq (2D + \ell_u) \cdot (2 D + \ell_w) \ . \]
This implies that $D$ must be at least $\min(\ell_u, \ell_w)/2$, which in turn is at least $\max(\ell_u, \ell_w)/4$, since the links are nearly-equilength.
Thus we can bound $\ell_w\leq 4D$ in Eqn.\ \ref{eq:dwu2} to obtain that $d_{wu}\leq 6D$.
\end{proof}
\begin{figure*}[ht]
	\begin{center}
		\includegraphics[scale=0.5]{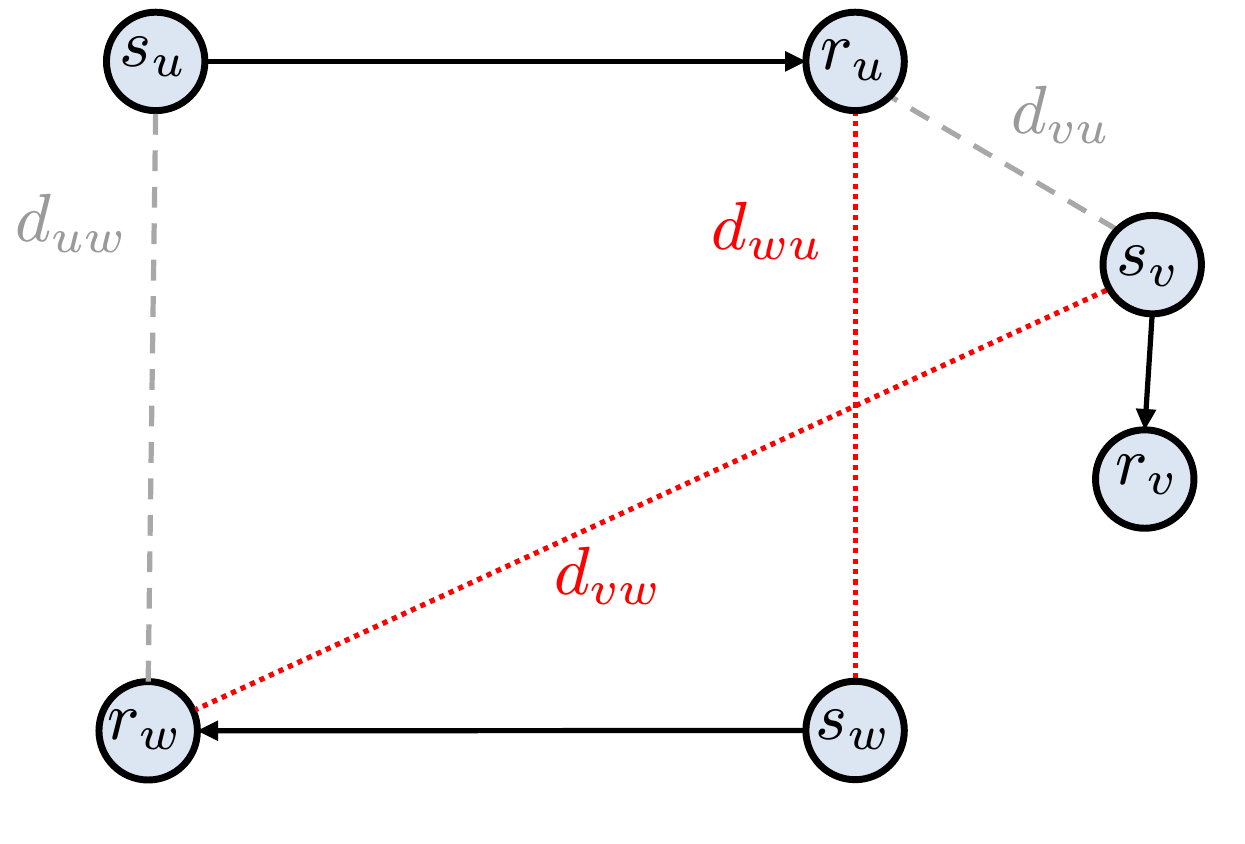}
	\end{center}
	\caption{Links $l_u,l_w$ and $l_w$ as used in the proofs of Propositions \ref{prop:geom1} and \ref{prop:geom2}. The distances $d_{vw}$ and $d_{wu}$ that are related to each other in the Proposition's statement are represented by red dotted lines. The gray dashed lines mark distances $d_{uw}$ and $d_{vu}$ that are used in the proofs as well.}\label{fig:geom2}
\end{figure*}

\footnote{MMH: Add an intuition here on the idea behind the lemma below?}

\begin{defn}
We say that links $l_v$ and $l_w$ are \emph{$t$-close} under
power assignment $\calP$ if, \[\max(a_{v}^\calP(w),a_{w}^\calP(v)) \ge t.\]
\end{defn}

For the rest of this section, denote $\hat{p} := \frac{1}{\min(1-p,p)}$.

\smallskip

\noindent \textbf{Lemma \ref{lem:lld}}\ 
\emph{
Let $p$ be a constant, $0 < p < 1$, $\tau$ be a parameter, $\tau \ge 1$, and $\Lambda =
(4 (2\beta \tau)^{1/\alpha})^{\hat{p}}$. 
Let $l_v$ be a link and 
let $Q$ be a 2-independent set of non-weak links in an arbitrary
metric space, that are both $\frac{1}{\tau}$-close to $l_v$
under power assignment $\powp$ and at least a $\Lambda$-factor longer than $l_v$.
Then, $|Q| = \BO(\log\log \Delta)$.
}
\smallskip


\begin{proof}
The set $Q$ consists of links that have at least one of the following properties:
\begin{enumerate}
\item a link can affect $l_v$ by at least $\frac{1}{\tau}$ under power assignment $\powp$
\item the link itself can be affected by $l_v$ by that amount
\end{enumerate}
We consider first the links with the first property. Consider a pair $l_w, l_{w'}$ in $Q$ that affect $l_v$ by at least $1/\tau$ under $\powp$, 
and suppose without loss of generality that $\ell_w \ge \ell_{w'}$.
Let $l_1$ be the shortest link in $Q$.
The affectance of $l_w$ on $l_v$ implies that
\[  c_v \left(\frac{\ell_w^p \ell_v^{1-p}}{d_{wv}}\right)^\alpha  
        \ge \frac{1}{\tau}\ , \]
which can be transformed to 
$d_{wv} \le \ell_w^p \ell_v^{1-p} (c_v \tau)^{1/\alpha}$, 
and similarly, $d_{w'v} \le \ell_{w'}^p \ell_v^{1-p} (c_v \tau)^{1/\alpha}$.
Recall that since $l_v$ is non-weak, $c_v \le 2\beta$.
By the triangular inequality, we have that 
\begin{align}\label{eq:dwprimew}
d_{w'w} & \le  d(s_{w'}, r_v) + d(r_v,s_w) + d(s_w, r_w) \\\nonumber
				& = 	 d_{w'v} + d_{wv} + \ell_w \\\nonumber
				& \le  2 \ell_w^p \ell_v^{1-p} (c_v\tau)^{1/\alpha} + \ell_{w}  \\\nonumber
				& \le  2 \ell_w^p \ell_v^{1-p} (2\beta\tau)^{1/\alpha} + \ell_{w}  \\\nonumber
				& \le  \ell_w^p \ell_1^{1-p} + \ell_{w} \le 2\ell_w \ ,
\end{align}
%
using that $\Lambda\ell_v \le \ell_1 \le \ell_w$.
Similarly, 
\begin{equation}
d_{w w'}  \le \ell_{w'} + \frac{1}{2}\ell_w^p \ell_1^{1-p} \ .
\label{eq:dwwprime}
\end{equation}
Applying 2-independence, on one hand, and multiplying Eqn.\ \ref{eq:dwprimew} and \ref{eq:dwwprime}, on the other, we obtain that
\begin{equation}
4 \ell_w \ell_{w'} \le d_{w'w} \cdot d_{w w'} 
   \le 2 \ell_{w'} \ell_w  + \ell_w^p \ell_1^{1-p} \cdot \ell_w\ ,
\label{eq:t2}
\end{equation}
Canceling a $2\ell_w$-factor, simplifying and rearranging,
we have that
\begin{equation}
\ell_w^p \ge \frac{2 \ell_{w'}}{\ell_1^{1-p}}\ .
\label{eq:two-length-rel}
\end{equation}
We label the links in $Q$ as $l_1, l_2, \ldots, l_{|Q|}$ in 
increasing order of length, and define $\lambda_i = \ell_i / \ell_1$.
By dividing both sides of Eqn.\ \ref{eq:two-length-rel} by $\ell_1^p$, 
we get that
\[ \lambda_{i+1}^{p} \ge 2 \lambda_i\ . \]
Then, $\lambda_2 \ge 2^{1/p}$ and 
by induction $\lambda_{t} \ge 2^{(1/p)^{t-1}}$.
Note that 
\[\Delta(Q)=\ell_{|Q|}/\ell_1 = \lambda_{|Q|}\geq 2^{(1/p)^{|Q|-1}},\]
 so
\[|Q|-1 \le \log_{1/p} \log_2 \Delta,\] and the claim follows.

The other case of links $l_w$ with $a_v(w) \ge 1/\tau$
is symmetric, with the roles of $p$ and $1-p$ switched, leading to a bound of $1 + \log_{1/(1-p)} \log_2 \Delta$. 

The case that a link can have both properties does not affect the asymptotic statement of the result.
\end{proof}

\medskip

\noindent \textbf{Proposition \ref{prop:l3bound}}\ 
$|L_3| \leq  2 \cdot 4^{\alpha} + 1$.
\smallskip

\begin{proof}
By Lemma \ref{lem:indep}, $L_3$ can be divided into $2 \cdot
4^{\alpha} + 1$ sets, each of which is $4$-independent.  For
contradiction, if $|L_3| > 2 \cdot 4^{\alpha} + 1$, then at least one
of these sets must be of size at least $2$.  Thus, there would be two
different links $l_x$ and $l_y$ that are members of $L_3$ and
are $4$-independent.

However, since $l_x, l_y \in L_3$, we can argue that 
\begin{align*}
d(x,  y) \overset{1}{\leq} \ell_x + d(r_x, r_y)   \overset{2}{\leq} \ell_x + D \overset{3}{\leq} \ell_x + 2 \ell_x \leq 3 \ell_x\ ,
\end{align*}
Explanation of numbered inequalities:
\begin{enumerate}
\item By triangle inequality.
\item Observing that both  $r_x$ and $r_y$ are in $B(s_v, D/2)$ (due to the definition of $L_3$) and using the triangle inequality.
\item Since $\ell_x = d(s_x, r_x) \geq D/2$ as $r_x \in B(r_v, D/2)$ (since $l_x \in L_3$) and $d(s_x, r_v) \geq D$ (by definition of $D$)
\end{enumerate}

We can similarly show that $d(y, x) \leq 3 \ell_y$. Then $d(x, y) \cdot d(y, x) \leq 9 \ell_x \ell_y$, contradicting $4$-independence.
\end{proof}

\end{document}